\documentclass[superscriptaddress,reprint,amsmath,amssymb,prl]{revtex4-1}
\usepackage{graphicx}
\usepackage{amsthm}
\newtheorem{theorem}{Proposition}
\newtheorem{definition}{Definition}

\begin{document}

\preprint{APS/123-QED}

\title{Nodal dynamics, not degree distributions, determine the \\
structural controllability of complex networks}
\author{Noah J.~Cowan}
\affiliation{Department of Mechanical Engineering, Johns Hopkins University, Baltimore, MD 21218}
%\email{Corresponding author: ncowan@jhu.edu}
\author{Erick J.~Chastain}
\affiliation{Department of Computer Science, Rutgers University, New Brunswick, NJ 08901}
\author{Daril A.~Vilhena}
\affiliation{Department of Biology, University of Washington, Seattle, WA 98105}
\author{James S.~Freudenberg}
\affiliation{Department of Electrical Engineering and Computer Science, University of Michigan, Ann Arbor, MI 48109}
\author{Carl T.~Bergstrom}
\affiliation{Department of Biology, University of Washington, Seattle, WA 98105}
\affiliation{Santa Fe Institute, 1399 Hyde Park Rd., Santa Fe, NM 87501}

\date{\today}

\begin{abstract}
Structural controllability has been proposed as an analytical framework for making predictions regarding the control of complex networks across myriad disciplines in the physical and life sciences (Liu et al., Nature:473(7346):167-173, 2011). Although the integration of control theory and network analysis is important, we argue that the application of the structural controllability framework to most if not all real-world networks leads to the conclusion that a single control input, applied to the power dominating set, is all that is needed for structural controllability. This result is consistent with the well-known fact that controllability and its dual observability are generic properties of systems. We argue that more important than issues of structural controllability are the questions of whether a system is almost uncontrollable, whether it is almost unobservable, and whether it possesses almost pole-zero cancellations.
\end{abstract}

\maketitle

\section{Introduction}

How can we control complex networks of dynamical systems
\cite{sorrentino2007controllability,duan2007complex,porfiri2008criteria,zhou2008pinning,li2004synchronization,xiang2007pinning,sorrentino2007effects,lozovanudetermining2010,lozovanuoptimal2010}?
Is it sufficient to control a few nodes, or are inputs needed at a
large fraction of the nodes in the network? Which nodes need to be
controlled? A recent paper \cite{liucontrollability2011} suggests that
we can address these problems using the concept of structural
controllability \cite{linstructural1974}, and in doing so we may be
able to forge new connections between control theory and complex
networks.  The two main results from this analysis are (1) that the
number of driver nodes, $N_D$, necessary to control a network is
determined by the network's degree distribution and (2) that $N_D$
tends to comprise a substantial fraction of the nodes in inhomogeneous
networks such as the real-world examples considered therein.

However, both conclusions hinge on a critical assumption of the model
in \cite{liucontrollability2011}: the results (implicitly) require
that the ``default'' structures of the dynamical systems at the nodes
of the network have infinite time constants. This modeling assumption
implies that, unless otherwise specified by a self-link in the
network, a node's state never changes absent influence from inbound
connections. However, the real networks considered
in~\cite{liucontrollability2011}---including food webs, power grids,
electronic circuits, regulatory networks, and neuronal
networks---typically manifest more general dynamics at each node,
i.e.\ they typically have finite time constants
\cite{lestasfundamental2010,marderprinciples1996,berrymancredible1995}.

With this assumption, the minimum number of independent control inputs
required to ensure a technical property known as \emph{structural
  controllability} \cite{linstructural1974} can be calculated for the
network, as described in \cite{liucontrollability2011}. The main
problem with the argument set forth in \cite{liucontrollability2011}
is not a technical one: indeed the assumptions therein are clear and
the mathematical results are correct. Then, why are the results
tenuous? Critically, structural controllability
\cite{linstructural1974} is premised on the idea that if the
\emph{nonzero} parameters in the mathematical model can be selected so
that the system is controllable (an elementary concept in control
theory; see for example \cite{rughlinear1996}), then the system will
be controllable for all parameters except a set of zero measure. That
is, if the system is controllable for one set of (initially nonzero)
parameters, then controllability is guaranteed generically for that
system. The results presented in \cite{liucontrollability2011} require
that a critical assumption be made \emph{before} applying the
structural controllability approach.  Specifically, it is assumed that
each node has an infinite time constant. As we shall see in the next
section, the assumption of an infinite time constant implies that a
certain parameter in the mathematical model of the system is equal to
zero, and therefore that term is off-limits as far as structural
controllability is concerned.  As one can imagine, any approach to
system analysis that only allows the modification of nonzero terms,
makes the results potentially quite sensitive to which terms are set
to zero in the first place. Indeed, if the infinite-time-constant
assumption is relaxed, and generic linear dynamics are ascribed to
each node, one obtains a categorically different result.  Indeed, we
show in this paper that all networks with finite-dimensional linear
dynamics (save a special set of parameters of zero measure) are
controllable with a single input. While mathematically true, such a
conclusion is neither reasonable nor practical for real-world
networks, and thus calls into question the general approach of
applying structural controllability in this way.

Assuming arbitrary (up to a set of measure zero) linear dynamics, we
show here that (1) a single time-dependent input is all that is needed
for structural controllability, and (2) that this input should be
applied to the power dominating set (PDS)
\cite{aazamiapproximation2007} of the network. Thus for many if not
all naturally occurring network systems, structural controllability
does not depend on degree distribution and can always be conferred
with a single control input.

\section{Modeling Networks for Control}
Large interconnected systems are commonly represented as complex networks \cite{strogatz2001exploring,newman2001clustering}. For many biological and physical networks, each node in the network corresponds to a dynamical system. Often, the dynamics of these nodes can be modeled by a system of ordinary differential equations \cite{chenpinning2007,wangpinning2002}:
\begin{equation}\label{eq:nodaldyn}
  \dot x_i = -p_i x_i + \sum_{k = 1}^N a_{ik} x_k(t) + \sum_{j=1}^P b_{ij} u_j(t),
\end{equation}
where $x_i$ is a state at node $i$, $N$ is the number of nodes, $P$ is
the number of inputs, and the $n^2$ elements $a_{ik}$ populate the
adjacency matrix.  Here, the term $-p_i x_i$ represents the intrinsic
dynamics at the node, absent external influences. The external inputs,
$u_j(t)$, enter the system through the coupling matrix
$\{b_{ij}\}$. For analyzing controllability, it is reasonable as a
first step to consider purely linear dynamics as shown in
Eq.~\eqref{eq:nodaldyn}---an approach clearly articulated and well
motivated by \cite{liucontrollability2011}.

Note that Eq.~\eqref{eq:nodaldyn} includes two terms in dynamics for
$x_i$, one related to the linearization of the intrinsic nodal
dynamics, namely $-p_i x_i$, and one related to a potential self link
in the model, namely $ a_{ii} x_i$, related to the network
topology. Although both terms are identical mathematically, they arise
from categorically different sources, and thus are not
interchangeable.

The term $-p_i$ is the \emph{pole} of the linear dynamical system at
each node, and $\tau_i = 1/p_i$ is the associated time
constant. Rewriting in terms of transfer functions, we have
\begin{equation}\label{cowanmodel}
 X_i(s) =G_i(s)  \left[ \sum_{k=1}^{N} a_{ik} X_k(s) + \sum_{j=1}^{P } b_{ij} U_j(s) \right], 
\end{equation}
where $X_i(s)$  and  $U_j(s)$ are the Laplace transforms of state $x_i(t)$ and input $u_j(t)$ respectively, and
\begin{equation*}
  G_i(s) = \frac{1}{s+p_i},
\end{equation*}
is the transfer function of node $i$. This formulation is useful because it suggests inclusion of more general linear dynamics: the transfer function, $G_i(s)$, can be replaced by any transfer function, of arbitrary order.

The dynamics proposed in \cite{liucontrollability2011} (see the
supplemental material therein) are identical to \eqref{cowanmodel},
except that $p_i\equiv 0$ for all $i$, namely $ G_i(s) = 1/(s +
0)$---a pure integrator. Written this way the simplifying assumption
of the model in \cite{liucontrollability2011} becomes clear: all
subsystems by default have an infinite time constants (that is, the
term $p_i=0$) unless such dynamics are explicitly included in the
network data set through nonzero diagonal terms, $a_{ii}\neq 0$, in
the adjacency matrix. 

However, infinite time constants at each node do not generally reflect
the dynamics of the physical and biological systems in Table 1 of
\cite{liucontrollability2011}.  Reproduction and mortality schedules
imply species-specific time constants in trophic networks. Molecular
products spontaneously degrade at different rates in protein
interaction networks and gene regulatory networks.  Absent synaptic
input, neuronal activity returns to baseline at cell-specific
rates. Indeed, most if not all systems in physics, biology, chemistry,
ecology, and engineering will have a linearization with a finite time
constant. Thus while the model in \cite{liucontrollability2011} does
not proscribe self-links, this approach does place the onus on the
modeler to ensure that any network representation includes such
self-links where appropriate to compensate for the omission of the
\emph{intrinsic} nodal dynamics that arise due to physical,
biological, or other processes that, generally speaking, have nothing
to do with network topology.

To see the consequences of including generic nodal dynamics on a
network's structural controllability, we first rewrite the network
dynamics in \eqref{cowanmodel} in state space form: 
\begin{equation}
  \label{eq:sys}
  \begin{aligned}
    \dot{\mathbf{x}}(t) &= \hat{A}  \mathbf{x}(t) + B  \mathbf{u}(t), \\
    \hat{A} &= \left[A - \mathrm{diag}(p_1,p_2,p_3,\ldots,p_N) \right],
 \end{aligned}
\end{equation}
where $A\in\mathbb{R}^{N\times N}$ is the adjacency matrix, and $B \in\mathbb{R}^{N\times P}$ is the input matrix. The vector $\mathbf{x}(t)\in\mathbb{R}^N$ is the vector of node states, and $\mathbf{u}(t)\in\mathbb{R}^P$ is the  input vector.

The system in Eq.~\eqref{eq:sys} is controllable if and only if the matrix
\begin{equation}
  \label{eq:controllability}
  \begin{bmatrix}  B, & \hat{A}  B, &  \cdots &\hat{A}^{N-1} B \end{bmatrix}
\end{equation}
 is full rank, a standard result in control theory \cite{rughlinear1996}. The system is said to be \emph{structurally controllable} if the nonzero weights in $\hat{A}$ and $B$ can be adjusted such that the matrix in Eq.~\eqref{eq:controllability} is full rank~\cite{linstructural1974}.

 In \cite{liucontrollability2011}, the minimum number of driver nodes,
 $N_D$, is defined as the minimum number of inputs---i.e.,\
 independent, user defined, time-varying functions---such that when
 injected into the network guarantee structural controllability. This
 formulation explicitly allows each independent input to be connected
 to multiple (and possibly all)\ nodes in the network
 \cite{liucontrollability2011,liupersonal2011}.

 The paper \cite{liucontrollability2011} solves this minimum input
 problem using an  application of graph-theoretic
 concepts; their basic approach is to identify the number of
 ``unmatched nodes'' after finding a so-called maximum matching of the
 graph. Details are provided in the supplemental material of
 \cite{liucontrollability2011}; note also the prior analysis wherein
 the maximum matching theorem seems first to have been proved
 \cite{commaultcharacterization2002}. We observe that one can recast
 the poles at $-p_i$ as (nonzero) self-links. But the set of all
 self-links $(i\rightarrow i)$ is itself a maximum matching; all nodes
 in the network are then matched nodes. This implies that the network
 can be controlled with a single input, i.e.\ $N_D=1$, which follows
 directly from the maximum matching proof in
 \cite{liucontrollability2011,slotinepersonal2011}.

 \section{Structural Controllability of Networks with General Linear
   Dynamics }
The following proposition provides a simple non-graph-theoretic proof that a ``control hub'' --- a  single driver node attached to \emph{all nodes} --- guarantees structural controllability with a single input.
\begin{theorem}\label{thm:hub}
  For any directed network with nodal dynamics in Eq.~\ref{cowanmodel}
  (or equivalently Eq.~\ref{eq:sys}), with $p_i\neq 0$ and/or
  $a_{ii}\neq 0$, $i=1,\ldots,N$, then $N_D=1$.
\end{theorem} 
\begin{proof}
Select $B=[1, 1, \ldots, 1]^T$ (that is, connect a single input to all nodes). Lin's structural controllability theorem \cite{linstructural1974} states that if the system is controllable for one choice of the nonzero system parameters, then it will be controllable for all parameters except a set of measure zero. So, we explicitly construct a parameter set that makes the system controllable. Keep $B$ as all ones, and choose $p_1,p_2,\ldots,p_N$ to be nonzero and distinct.  Zero out all the network edges (i.e.\ nullify the adjacency matrix, $A=0$). The system matrix $\hat{A}$ is now a diagonal matrix with distinct eigenvalues. Controllability of $(\hat{A},B)$ follows by inspection. Thus, the system is structurally controllable and $N_D=1$.
\end{proof}

By contrast the paper \cite{liucontrollability2011} reported that for
real-world networks, the minimum number of driver nodes $N_D$ is
strongly influenced by the sparseness and homogeneity of the network,
as measured by the \emph{degree distribution},
$P(k_{\mathrm{in}},k_{\mathrm{out}})$ (see
\cite{liucontrollability2011} for more details). Why did
\cite{liucontrollability2011} arrive at such different conclusions?
Critically, the application of structural controllability does not
consider variations in system parameters that are \emph{a priori} zero
\cite{linstructural1974}. So, for example, if a link $i\rightarrow j$
is absent, then $a_{ij}\equiv 0$. The original paper
\cite{liucontrollability2011} allows for self-links but by default
does not include them. Further, the framework set forth in
\cite{liucontrollability2011} assumes $p_i=0$ (infinite time
constant), and the network datasets in Table 1 of
\cite{liucontrollability2011} do not include self-links to correct for
this. Therefore, upon inclusion of first-order self dynamics,
essentially all real networks are structurally controllable with
$N_D=1$, irrespective of network topology.

In the case that the network topology does not explicitly contain self
links, the consequence of ascribing pure integrator dynamics ($p_i=0$)
to each node is categorical: the system is necessarily unstable. This
is because the sum of the eigenvalues is given by the trace of the
system matrix, which, in this case, would be
$\mathrm{trace}(\hat{A})=0$, since there are zeros on the entire
diagonal. This would imply that it is impossible to have any stable
eigenvalues (negative real parts) without also having unstable ones
(positive real parts), so that their sum is zero. Therefore, such a
network of integrators must be purely oscillatory or unstable, and
cannot be asymptotically stable. Therefore, assuming pure integrators
at each node, and no explicit self-links in the adjacency matrix,
precludes passive stability which many natural systems enjoy.

Have we taken the point about generic nodal dynamics too far? It may
be desirable to model and control a network on a timescale that is
faster than the dynamics of the intrinsic nodal dynamics. We concede
that in such cases, it may be reasonable to treat the nodal dynamics
as pure integrators (systems with infinite time constants). However,
we argue that structural controllability may not be appropriate for
addressing these nuanced modeling issues. An essential feature of
structural controllability is that no importance is assigned to
\emph{specific values} for the non-zero terms in the dynamics. Values
are treated as either zero or not zero; there is no in-between. Thus,
the choice of whether to zero out the self-loop terms \emph{a priori}
is a subtle modeling issue that should take into account the
\emph{emergent timescales of the entire network}. Therefore, we
contend that model reduction \cite{mooreprincipal1981}---which is
essential for controller design---should be treated at the level of
the entire network dynamics rather that at the level of individual
nodes: indeed the timescales relevant for control are an emergent
property of the \emph{system} dynamics, and not strictly a feature of
one node or another. With this in mind, we find that the tool of
structural controllability---which is premised on a notion of generic
parameters---is best suited to generic modeling assumptions. In this
case this means assuming $p_i\neq 0$, $i=1,\ldots N$.

Above, we argue that structural controllability of complex networks
depends on the dynamics at each node, and that only a single time
varying input is required. Two questions remain: (1) How sensitive is
structural controllability to the dimension of the state space for
each node? (2) Where should we inject the $N_D$ independent time
inputs into the network, i.e.\ what is the minimum number of nodes of
the network to which the input must be connected? Proposition
\ref{thm:hub} explicitly depends on treating first order nodal
dynamics as ``self loops'' in the network. Below we offer a more
general treatment for arbitrary (linear) nodal dynamics that addresses
both questions above. See Figure~\ref{fig:pds}.

  \begin{figure}[htb]
    \centering
    \includegraphics{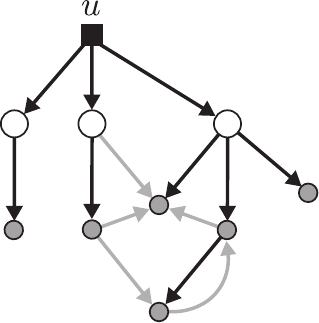}
    \caption{Given a network, the PDS (large white circles) is the smallest set of nodes such that all other nodes (smaller grey circles) are downstream of them. Any network, with arbitrary (and possibly different) order finite-dimensional linear dynamics at each node is structurally controllable from a single driver node (black square) tied to the PDS as shown. See Proposition~\ref{thm:pds}. The edges in the \emph{structural control network} are part of a minimum spanning tree (black edges, although this choice of edges, and indeed the PDS, is not necessarily unique). }
    \label{fig:pds}
  \end{figure}

Given a directed graph, a PDS is, by definition, the smallest set of
nodes such that all other nodes are downstream of at least one node in
the PDS. Obviously, controllability requires connecting the input(s)
at least to this set; below we show that structural controllability is
generically achieved by connecting a single input to the PDS. Before
doing this, we need one definition:

\begin{definition}
  Suppose that there are $K$ nodes in the PDS. Attach a single control
  input, $u$, to this set via a control node. Augment the graph with
  this control node and add the $K$ edges that connect it to the
  PDS. Then, all nodes are downstream of the input $u$ (i.e.\ the
  control node is now the PDS of the augmented graph). Define the
  \textbf{structural control network} as an acyclic directed graph given
  by a directed spanning tree that starts at $u$ and visits all nodes.
\end{definition}

We now state the main result.

\begin{theorem}\label{thm:pds}
  Consider the nodal dynamics in \eqref{cowanmodel}, with $G_i(s)$ an arbitrary, proper, rational transfer function \cite{rughlinear1996} of the form
  \begin{equation*}
    G_i(s) = \frac{n_i(s)}{d_i(s)},
  \end{equation*}
  where, $n_i(s)$ and $d_i(s)$ are assumed to be generic polynomials
  (all coefficients up to the order of the polynomial are assumed to
  be nonzero) of finite but arbitrary order in $s$. Then, the network
  is structurally controllable with one ($N_D=1$) independent input,
  connected to the PDS.
\end{theorem}
\begin{proof}
  Using the structural controllability argument, we are free to modify
  any nonzero parameters; if the system is controllable for one set of
  parameters, it will be generically controllable.

  So, zero out all edges that are not in the structural control
  network and set all those in the structural control network to 1; if
  this process results in a controllable system, as we now show it
  does, then the system will be controllable generically. 

  All nodes in this structural control network are still downstream of
  $u$, but now there are no cycles. Since the structural control
  network is a minimum spanning tree, there is exactly one path
  between $u$ and any specific node, $i$. Let $\mathcal{J}_i$ denote
  the set of nodes along the path from $u$ to node $i$ in the
  structural control network. Then transfer function from $u$ to any
  given node is simply the product of the transfer functions along the
  path from $u$ to the node:
  \begin{equation}
    \label{eq:1}
    H_j(s) = \frac{X_j(s)}{U(s)} = \prod_{k\in\mathcal{J}_i} G_k(s),
    \quad j=1, 2,\ldots N.
  \end{equation}

  Since we may freely adjust the polynomial coefficients in the
  denominator terms, we do so to ensure there are no repeated poles in
  the entire network and similarly adjust the numerator coefficients
  to ensure no pole-zero cancellations along any path in the
  structural control network.  Since there are no pole-zero
  cancellations, and all poles in the network are unique, a minimal
  realization of the $N\times 1$ transfer function $[H_1(s), H_2(s),
  \ldots H_N(s)]^T$ must contain exactly one eigenvalue for each pole
  of the network. It is obvious that the minimal realization requires
  no more eigenvalues than that. The number of eigenvalues in the
  minimal realization is equivalent to the number of eigenvalues that
  are both controllable and observable. Thus all states are
  controllable for this parameter set and, by the' structural
  controllability theorem \cite{linstructural1974}, the network is
  structurally controllable.
\end{proof}

For first-order nodal dynamics, our main result is not substantively
different from those presented for discrete time and finite state
systems in \cite{sundaramcontrol2010,sundaramdistributed2008}. They
show that networks with nontrivial nodal dynamics are structurally
observable with a single output node and structurally controllable
with a single input node. Our modest generalization to arbitrary-order
nodal dynamics is at best incremental over their work. Indeed, the
main contribution of our paper lies not so much in any technical
advance as it does in providing a timely clarification of
\cite{liucontrollability2011}.

\section{Simple Example: A Food Web}
To illustrate the ideas of this paper, consider a simple food web comprising one predator and one prey species. Let $H$ denote the number of herbivores (prey) and $C$ denote the carnivores (predators).  We begin by noting that historically, the classic models of predator--prey dynamics  \cite{lotkaanalytical1920} take the form\begin{equation}
\label{eq:lotka}
  \begin{aligned}
  \dot C &=  -\gamma C + \epsilon C H\\
    \dot H &= \alpha H  - \beta H C
 \end{aligned}
 \implies
  \begin{aligned}
  \dot C &=  C(-\gamma + \epsilon H)\\
    \dot H &=  H(\alpha - \beta C),
 \end{aligned}
\end{equation}
where $\alpha,\beta,\gamma,\epsilon>0$. Linearizing these dynamics about the nontrivial equilibrium $(C^*,H^*)=(\frac{\alpha}{\beta},\frac{\gamma}{\epsilon})$ this model has the following local dynamics:
\begin{equation*}
  \dot x =
  \begin{bmatrix}
    0 & \frac{\alpha \epsilon}{\beta} \\
    -\frac{\beta \gamma}{\epsilon}  & 0
  \end{bmatrix}
  x
\end{equation*}
where $x=(\delta C, \delta H)$ is the vector of small displacements
relative to the equilibrium $(C^*,H^*)$. Note that the linearized food
web is fully connected (whereas \cite{liucontrollability2011} include
a nonzero edge for ``$C$ eats $H$'' but not for ``$C$ is eaten by $H$"
in their treatment of a trophic networks). Also, note that this
linearized system has infinite time constants at the nodes, i.e.,\
zero values along the diagonal. Thus these early models do not include
the finite time constants that we argue are so important to system
dynamics. Later work remedied this omission; the early models such as
Eq.~\eqref{eq:lotka} did not include terms that researchers
subsequently found to be essential for modeling real biological
systems, such as saturation effects arising from resource limitations
\cite{berrymancredible1995}. Including these additional terms leads to
a $2\times2$ system matrix that is fully populated with (generically)
nonzero terms on and off the diagonal. This implies that the
resulting linearization features finite time constants at each of the
nodes, and the network is fully connected. That is, where structural
controllability is concerned, taking into account the full dynamics of
a food web leads inescapably to the conclusion this system should be
controllable with a single input.

\section{Discussion}

Recently, it was reported that sparse inhomogenous networks require
distinct controllers for a large fraction of the nodes to attain
structural controllability \cite{liucontrollability2011}. We argue
that these results are a consequence of assuming a special structure
for the dynamics at each node: each node is treated as a pure
integrator. In the application of the model set forth in
\cite{liucontrollability2011} to the real networks considered therein,
each node is assumed to have an infinite time constant. In this paper,
we show that (1) for generic, arbitrary-order nodal dynamics,
structural controllability can be achieved with a single time-varying
input, and (2) that input should be attached to a PDS.

The property of a system being controllable has two significant interpretations in control theory. First, if a system is controllable then it is possible to find an input to transfer any initial state to any final state in finite time. Second, if a system is controllable then it is possible to apply a control signal consisting of a linear combination of the states that changes the dynamics arbitrarily. In particular, it is possible to stabilize an unstable system, a necessary design goal in engineering problems. Such a control signal is termed state feedback.

It is important to note what the first definition of controllability leaves out. For example, unless the final state is an equilibrium, the state will not remain there, but will move away. In many engineering applications, it is important to find an input that will both stabilize a system and hold a specified linear combination (or set of linear combinations) of states at desired constant values. This is referred to as the problem of setpoint tracking, and requires that the system be controllable (so that a stabilizing control input may  be found) and that there are at least as many independent control inputs as there are linear combinations of states to be held at desired setpoints \cite{kwakernaaklinear1972}. Hence we see that although one input may suffice to achieve controllability of an arbitrary number of state variables, in fact the number of inputs limits the number of setpoints that may be specified.

The property of controllability is generically present in a system, and thus in practice it is more important to know not whether a system is controllable, but whether it is almost uncontrollable. In the latter case, the control input used to drive the state to its desired value, or to achieve the desired dynamics, may be excessively large. Hence there is a need for tests---such as those based on the control Gramian \cite{rughlinear1996}---to determine what states are almost uncontrollable. In practice these are then treated as though they were indeed uncontrollable to avoid the excessively large inputs required to control them.

A more subtle problem arises with the second use of the controllability property. In practice, it is rarely possible to measure all the states of the system required for the control signal used to alter the dynamics of the system. Instead, the control signal is based on estimates of the states obtained by processing those states (or linear combinations of states) that are measurable. A system is said to be observable if it is possible to estimate the states using only the available outputs \cite{rughlinear1996}. As is the case with controllability, the property of observability is generically present, and it is necessary to determine whether states are almost unobservable.

States that are either uncontrollable or unobservable do not influence the input--output relation of a system, and cannot themselves be influenced by a control input signal based on output measurements. Such systems are characterized by a pole (an eigenvalue of the matrix $\hat{A}$) that does not appear in the transfer function due to being canceled by a zero of the transfer function having the same value. If the system is almost uncontrollable or almost unobservable, then the transfer function will have a  zero very near to a pole. In this case, it is possible to design a control signal based on state estimates. However, it may be shown using the theory of fundamental design limitations \cite{freudenbergright1985,loozetradeoffs2010} that the resulting feedback control system will necessarily have a very small stability margin, and be sensitive to disturbances and parameter variations. Often, the solution to this problem requires the introduction of additional control inputs or additional measurements.

In conclusion, the property of controllability, although important, is by no means sufficient to assure a well behaved control problem. One might expect this to be true since the property is generically present, as is the property of observability. The more relevant questions are thus whether the system is almost uncontrollable, almost unobservable, or possesses almost pole--zero cancellations.

\section{Acknowledgments}
The authors thank Jean-Jacques Slotine and Yang-Yu Liu for their
generous assistance in developing the arguments presented here. We
further thank Jessy Grizzle, Tom Daniel, Eric Fortune, and Andrew
Lamperski for a number of useful discussions.  Thanks to Shreyas
Sundaram for providing a critical evaluation of the manuscript and for
pointing out a number of relevant references. This material is based
upon work supported by the National Science Foundation under Grants
0845749 (NJC) and SBE-0915005 (CTB), and by US NIGMS MIDAS Center for
Communicable Disease Dynamics 1U54GM088588 (CTB) at Harvard
University.

\bibliography{netcontrol}

\end{document}